\documentclass[smallextended]{svjour3}

\usepackage{mathptmx}
\usepackage{xcolor}
\usepackage{graphicx}
\usepackage{amsmath}
\usepackage{amssymb}
\usepackage[numbers]{natbib}

\newcommand{\real}{\mathbb{R}}
\newcommand{\complex}{\mathbb{C}}

\newcommand{\tg}{\widetilde{g}}
\newcommand{\tu}{\widetilde{u}}
\newcommand{\tU}{\widetilde{U}}
\newcommand{\M}[1]{\left({#1}\right)}
\newcommand{\Mb}[1]{\left[{#1}\right]}
\newcommand{\Mcb}[1]{\left\{{#1}\right\}}
\newcommand{\Ma}[1]{\left\langle{#1}\right\rangle}

\newcommand\eq[1] {(\ref{#1})}

\newcommand{\Ga}{\alpha}
\newcommand{\Gb}{\beta}
\newcommand{\Gd}{\delta}

\def\Bc{{\bf c}}

\begin{document}
\title{Mathematical analysis of the two dimensional\\ active exterior cloaking in the quasistatic regime}
\titlerunning{Two dimensional active exterior cloak in quasistatics}
\author{Fernando~Guevara~Vasquez \and Graeme~W.~Milton \and Daniel~Onofrei}
\authorrunning{F. Guevara Vasquez, G. W. Milton and D. Onofrei}
\institute{F. Guevara Vasquez \and G. W. Milton \at 
 Department of Mathematics, University of Utah, Salt Lake
City UT 84112, USA.\\
\email{fguevara@math.utah.edu, milton@math.utah.edu}
\and
 D. Onofrei \at Department of Mathematics, University of Houston, Houston TX 77004,
USA.\\
\email{onofrei@math.uh.edu}
}
\maketitle
\begin{abstract}
 We design a device that generates fields canceling out a known probing
 field inside a region to be cloaked while generating very small fields
 far away from the device. The fields we consider satisfy the Laplace
 equation, but the approach remains valid in the quasistatic regime 
 in a homogeneous medium. We start by relating the problem of
 designing an exterior cloak in the quasistatic regime to the classic
 problem of approximating a harmonic function with harmonic polynomials.
 An explicit polynomial solution to the problem was given earlier in
 [Phys.  Rev. Lett. 103 (2009), 073901]. Here we show convergence of the
 device field to the field
 needed to perfectly cloak an object. The convergence region limits the
 size of the cloaked region, and the size and position
 of the device.
\end{abstract}

\keywords{Cloaking \and Laplace equation \and Quasistatics \and harmonic polynomial
approximation.}
\subclass{
31A05  
\and
35J05  
\and
30E10  
}

\section{Introduction}

Cloaking -- preventing detection of objects from a probing field -- has
been the subject of many recent studies, see e.g. the reviews
\cite{Alu:2008:PMC,Greenleaf:2009:IIP}.  A cloak can be {\em active} or
{\em passive} depending on whether active sources are needed to
maintain the cloak. A cloak is said to be {\em interior} if it
completely surrounds the object to be hidden and {\em exterior}
otherwise.

One approach to obtain {\em passive interior} cloaks is to exploit the
invariance of the governing equations (e.g. Laplace, Helmholtz, Maxwell
equations, $\ldots$) to coordinate transformations.  This approach was
introduced in
\cite{Greenleaf:2003:ACC,Pendry:2006:CEM,Leonhardt:2006:NCI,Leonhardt:2006:OCM,Chen:2007:ACT,Greenleaf:2009:IIP}
(see also references in \cite{Alu:2008:PMC,Greenleaf:2009:IIP}) and is based on
ideas first observed in \cite{Dolin:1961:PCT}. Although transformation based
cloaking is set on solid mathematical grounds and has been demonstrated
experimentally in a variety of physical settings, the cloaks generated with
this approach require materials with extreme properties that are usually
approximated using specially designed metamaterials. Unfortunately
metamaterials used in electromagnetic transformation based cloaking are typically very
dispersive, meaning that the cloak operates only in a narrow band of
frequencies. Also losses in the cloak material generate heat that can make the
object detectable using infrared. Some recent results in generating broadband
low-loss metamaterials have been obtained in \cite{Smolyaninov:2009:AME}. In an
effort to overcome the shortcomings of transformation based cloaks, various
regularizations have been proposed (see \cite{Kohn:2010:CCV} and references
therein). 

Other {\em passive interior} cloaking methods include  plasmonic
cloaking (see \cite{Alu:2008:PMC} and references therein).  Cloaking
methods that are {\em passive} and {\em exterior} include cloaking with
complementary media \cite{Lai:2009:CMI}, cloaking by anomalous
resonances \cite{Milton:2006:CEA,Nicorovici:2007:OCT,Milton:2008:SFG}
and plasmonic cloaking \cite{Silvereinha:2008:CMA}.

An example of an {\em active interior} cloak appears in
\cite{Miller:2007:PC} and uses sources continuously distributed over a
closed surface surrounding the cloaked region in order to cancel out the
incident field inside the cloaked region.

Here we focus on an {\em active exterior} cloak for the 2D Laplace
equation \cite{Guevara:2009:AEC}, which can be easily adapted to 2D
quasistatics in a homogeneous medium. This scheme assumes the incident
or probing field is known and uses one active source (cloaking device)
to cancel the incident field in the cloaked region with no significant
perturbation in the far field. Thus an object inside the cloaked region
interacts very little with the probing field and becomes harder to
detect. Active exterior cloaking has been extended to the 2D Helmholtz
equation in \cite{Guevara:2009:BEC,Guevara:2011:ECA} and to the 3D
Helmholtz equation in \cite{Guevara:2011:TEA}. Our approach assumes a
homogeneous background medium and requires three (resp. four) devices or
antennas to construct a cloak for the 2D (resp.  3D) Helmholtz equation.

Our goal here is to rigorously justify the quasistatic cloaking method
of \cite{Guevara:2009:AEC}. Quasistatics refers to the Maxwell or
Helmholtz equations in the long wavelength limit, where the governing
equation is the Laplace equation. We start by describing the cloak setup
in Section~\ref{sec:setup}. Then in Section~\ref{sec:existence},  we
prove the existence of a solution for the 2D quasistatic active exterior
cloak, based on a classic harmonic approximation result due to Walsh
(see e.g. \cite{Gardiner:1995:HA}). Unfortunately the existence proof is
not constructive. We proposed a candidate constructive solution without proof in
\cite{Guevara:2009:AEC}, supported by numerical experiments. In
Section~\ref{sec:const} we give the arguments behind this solution and
prove that it does indeed solve the active exterior cloaking problem.

\section{Cloak setup and device requirements}
\label{sec:setup}
Three regions in $\real^2$ are needed to describe our cloak setup: the
region to be cloaked, the cloaking device, and the observation region.
See Figure~\ref{fig:kelvin} (left) for an example setup.  The main idea
of our cloaking method is to cancel out an (assumed known) incident
field $u_0$ inside the cloaked region while perturbing the far field
only slightly. Thus the total field inside the
cloaked region is practically zero and the scattered field from any
objects inside the cloaked region is reduced significantly. 

Here we consider the conductivity equation with conductivity one and a
harmonic incident field $u_0$ (i.e. $\Delta u_0 =0$). Without loss of
generality, we take as cloaked region the disk $B(\Bc,a) \subset
\real^2$, centered at $\Bc = (p,0) \in \real^2$, $p>0$, and with radius
$a>0$.  As in \cite{Guevara:2009:AEC}, we consider one cloaking device located
inside $B(0,\Gd)$, with $\Gd \ll 1$. The device generates a field $u$,
harmonic outside $B(0,\Gd)$. In order to cloak objects the device field
$u$ needs to satisfy the following requirements.  
\begin{enumerate}
  \item The total field $u+u_0$ in the cloaked region $B(\Bc,a)$ is very
  small.
  \item The device field $u$ is very small far away from the device,
  e.g. in the observation region $\real^2 \setminus B(0,R)$ for a large
   $R>0$.
 \end{enumerate}
In order for the device to be {\em exterior} to the cloaked region, we
must have
\begin{equation}
 p > a + \Gd.
 \label{eq:const1}
\end{equation}
Also the observation radius $R$ needs to be large enough to contain both
the device and the cloaked region:
\begin{equation}
 R>a+p.
 \label{eq:const2}
\end{equation}

\begin{figure}
 \begin{center}
  \includegraphics[width=0.9\textwidth]{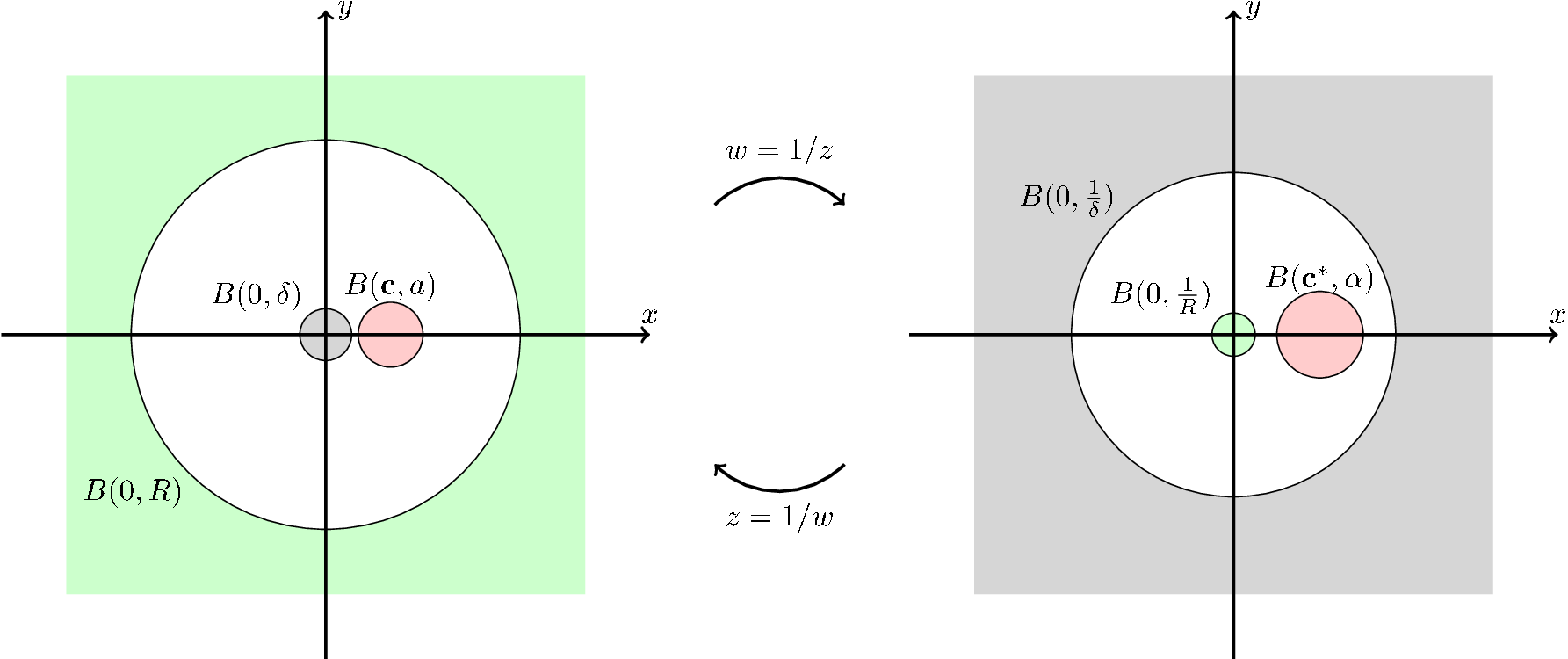}
 \end{center}
 \caption{The effect of the inversion (Kelvin) transform $w=1/z$ on the cloak
 geometry. The cloaked region is in red and the device sources are all
 contained in the gray disk. The green region is the observation region,
 where the device field must be very small to avoid detection.}
 \label{fig:kelvin}
\end{figure}

\section{Cloak existence}
\label{sec:existence}
The existence of a device field $u$ having the desired cloaking
properties to within a tolerance $\epsilon$ is stated in the next
theorem. 
\begin{theorem}
\label{thm-1}
Let $\epsilon>0$  be an arbitrarily small parameter. Let also $a>0$,
$\Bc=(p,0)$, $p>0$ and $R$ satisfy the inequalities \eqref{eq:const1} and
\eqref{eq:const2}.  Then for a harmonic incident field $u_0$, there are
functions $g_0:\real^2\to \real$ and $u:\real^2\rightarrow \real$ such that
 \begin{equation}
 \begin{aligned}
  &\left\{
  \begin{aligned}
   \Delta u &= 0, ~\text{in $\real^2 \setminus \overline{B(0,\delta)}$},\\
   u & = g_0, ~\text{on $\partial B(0,\delta)$},\\
  \end{aligned}
  \right.\\
  &\text{with $|u|  < \epsilon$ in $\real^2 \setminus B(0,R)$ and $|u+u_0| <
 \epsilon$ in $B(\Bc,a)$.}
 \end{aligned}
 \label{eq:clk}
 \end{equation}
\end{theorem}

The main idea of the proof of Theorem~\ref{thm-1} is to relate active
exterior cloaking to the problem of approximating harmonic functions
with harmonic polynomials. We rely on the following classic result.
\begin{lemma}[Walsh, see e.g. \cite{Gardiner:1995:HA}, page 8]
 \label{lem:walsh}
 Let $K$ be a compact set in $\real^2$ such that $\real^2 \setminus K$
 is connected. Then for each function $w$ harmonic on an open set
 containing $K$ and for any $\epsilon >0$, there is a harmonic
 polynomial $q$ for which $|w-q| < \epsilon$ on $K$.
\end{lemma}
We  can now proceed with the proof of Theorem~\ref{thm-1}.
\begin{proof}
It is convenient to use complex numbers $z = x + iy$ to represent points
$(x,y) \in \real^2$.
By applying the inversion (Kelvin) transformation $w = 
1/z$, the geometry of the problem transforms as in Table~\ref{tab:kelvin}.
(see also Figure~\ref{fig:kelvin}).

\begin{table}
 \begin{center}
 \renewcommand{\arraystretch}{1.2}
 \begin{tabular}{l|l|p{18em}}
  Region & $z$ plane & $w=1/z$ plane\\\hline
  Cloaking device &  $B(0,\delta)$ & $\real^2 \setminus
 B(0,1/\delta)$\\\hline
  Cloaked region & $B(\Bc,a)$ & $B(\Bc^*,\alpha)$ with
 $\Bc^* = (\beta,0)$, $\alpha = a/|p^2-a^2|$ and $\beta = p/(p^2-a^2)$\\\hline
  Observation region & $\real^2 \setminus B(0,R)$ & $B(0,1/R)$\\
 \end{tabular}
 \end{center}
 \caption{The different regions in our cloak setup and how they are mapped by
 the inversion (Kelvin) transformation.}
 \label{tab:kelvin}
\end{table}

Thus the cloaking problem \eqref{eq:clk} is equivalent to
finding functions $\tg_0$ and $\tu$ for which
\begin{equation}
 \begin{aligned}
 &\left\{
 \begin{aligned}
 \Delta \tu &= 0, ~\text{in $B(0,1/\delta)$},\\
 \tu &= \tg_0, ~\text{on $\partial B(0,1/\delta)$},
 \end{aligned}
 \right.\\
 &\text{with $|\tu| < \epsilon$ on $\overline{B(0,1/R)}$ and
$|\tu+\tu_0|<\epsilon$ on $B(\Bc^*,\alpha)$.}
 \end{aligned}
 \label{eq:kclk}
\end{equation}
 Here $\epsilon$ is as
in the statement of the theorem, $\tg_0(z) = g_0(1/z)$ and the function
$\tu_0(z) = u_0(1/z)$ is harmonic on $\real^2 \setminus \{0\}$.

Let $\tU_0$ denote the analytic extension of $\tu_0$
in $B(\Bc^*,\alpha)$, obtained by adding $i$ times its harmonic
conjugate. Notice that since $\tU_0$ is analytic, it can be arbitrarily
well approximated by a polynomial, e.g. a truncation of the power
series of $\tU_0$. Therefore, there is a polynomial $Q_0$ such that
\begin{equation}
 |\tU_0 - Q_0| < \epsilon /2, ~~\text{on $\overline{B(\Bc^*,\alpha)}$}.
 \label{3'''}
\end{equation}
 For $\tu_0$ this means that
\begin{equation}
 |\tu_0 - q_0| < \epsilon/2, ~~\text{on $\overline{B(\Bc^*,\alpha)}$},
 \label{3''}
\end{equation}
where $q_0$ is the real part of $Q_0$. Thus we may solve \eqref{eq:kclk} by
first approximating the (inverted) incident field $\tu_0$ by $q_0$ and then
studying the following problem
\begin{equation}
 \begin{aligned}
 &\left\{
 \begin{aligned}
 \Delta \tu &= 0, ~\text{in $B(0,1/\delta)$},\\
 \tu &= \tg_0, ~\text{on $\partial B(0,1/\delta)$},
 \end{aligned}
 \right.\\
 &\text{with $|\tu| < \epsilon$ on $\overline{B(0,1/R)}$ and
$|\tu+q_0|<\epsilon/2$ on $B(\Bc^*,\alpha)$.}
 \end{aligned}
 \label{eq:kclkpoly}
\end{equation}

After inversion, the conditions \eqref{eq:const1} and \eqref{eq:const2}
necessary for having an exterior cloak become
\begin{equation}
 \begin{aligned}
  1/R&<\Gb-\Ga, ~~ \text{(the two disks
$B(0,1/R)$ and $B(\Bc^*,\alpha)$ do not touch), and
 }\\
 \Gb+\Ga&<1/\Gd,~~ \text{(the two disks $B(0,1/\delta)$ and
 $B(\Bc^*,\alpha)$ do not touch).}
 \label{0'}
 \end{aligned}
\end{equation}
Therefore, there exists $0<\xi\ll 1$ such that 
\begin{equation}
\frac{1}{R}+\xi<\Gb-\Ga-\xi.
\label{6}
\end{equation} 
We can now apply Lemma~\ref{lem:walsh} to the compact set
$K=\overline{B(0,1/R)} \cup \overline{B(\Bc^*,\alpha)}$ (which has a connected complement by virtue of \eqref{0'}) and the function 
\begin{equation}
 w = \begin{cases}
  0 & \text{in $B(0,\frac{1}{R} + \xi)$,}\\
  -q_0 & \text{in $B(\Bc^*,\alpha+\xi)$,}
  \end{cases}
\end{equation}
which is a harmonic function in the open set $B(0,\frac{1}{R} + \xi)
\cup B(\Bc^*,\alpha+\xi)$ (a set containing $K$). We obtain that there
exists a harmonic polynomial $q$ such that $|q-w| < \epsilon/2$ on $K$.
A solution to \eqref{eq:kclkpoly} is then given by $\tu = q$ and $\tg_0
= q$ on $\partial B(0,1/\delta)$. This implies the statement of the
theorem.
\end{proof}

\begin{remark} 
We assumed throughout this section that the incident field $u_0$ is
harmonic on $\real^2$. This corresponds to a source located at infinity.
Recall our method relies on approximating the Kelvin transformed
analytic extension of the incident field $\tU_0$ inside the Kelvin
transformed cloaked region $B(\Bc^*,\alpha)$ by a polynomial $Q_0$ (see
\eqref{eq:approx1}). This approximation only requires analyticity of
$U_0$ inside the cloaked region $B(\Bc,a)$.  Hence the results of this
section and the construction of Section~\ref{sec:const} below generalize
easily to the case where the incident field $u_0$ is harmonic inside the
observation region $B(0,R)$. This is the case where the sources
generating the incident field are outside the observation region but not
necessarily located at infinity.
\end{remark}

\begin{remark}
Clearly, Theorem~\ref{thm-1} also holds when the device and cloaked
region are not necessarily disks. The only requirements are that they be
bounded, disjoint and that the complement of their union be connected (see
Lemma~\ref{lem:walsh}).
\end{remark}

\section{A constructive solution for active cloaking}
\label{sec:const}

Although mathematically rigorous, the existence result of Theorem
\ref{thm-1} does not give an explicit expression for the potential
required at the active device (antenna). To give an explicit harmonic
solution to problem \eq{eq:clk}, we first simplify the problem in
Section~\ref{sec:simp}. Then we give a candidate solution to the
simplified problem in Section~\ref{sec:lag}, in the form of a Lagrange
interpolation polynomial. A better solution is constructed in
Section~\ref{sec:her} by averaging several Lagrange interpolation
polynomials. The resulting polynomial turns out to be a Hermite
interpolation polynomial. Then in Section~\ref{sec:conv} we show that
this Hermite interpolation polynomial solves \eqref{eq:kclk} (and thus
the cloaking problem \eqref{eq:clk}) provided its degree is sufficiently
large. This convergence study reveals constraints on the size of the cloaked
region and the device that are due to the particular solution we
construct.

\subsection{Simplifying the problem}
\label{sec:simp}
In the proof of Theorem~\ref{thm-1}, we related the cloaking problem
\eqref{eq:clk} to the problem of approximating a polynomial $Q_0$ with
an analytic function $V$ such that for some $\epsilon>0$,
\begin{equation}
 |V| < \epsilon ~\text{in $B(0,1/R)$ and}~ |V+Q_0|<\epsilon ~\text{in
 $B(\Bc^*,\alpha)$}.
 \label{eq:approx1}
\end{equation}
Now consider the problem of finding an analytic function $W$ such that
for some $\epsilon'>0$,
\begin{equation}
 |1-W| < \epsilon' ~\text{in $B(0,1/R)$ and}~ |W|< \epsilon'~\text{in
 $B(\Bc^*,\alpha)$}.
 \label{eq:approx2}
\end{equation}
Assuming we can find an approximant $W$ in \eqref{eq:approx2} with
$\epsilon' = \epsilon/M$ and 
\begin{equation}
 M = \sup_{z \in B(\Bc^*,\alpha) \cup B(0,1/R)} |Q_0(z)|,
\end{equation}
a solution to \eqref{eq:approx1} is then $V=-Q_0 (1-W)$, which is analytic
because the product of two analytic functions is analytic.

For illustration purposes we fast forward to Figure~\ref{fig:ensavg},
where we give an example of a function $W$ with the approximation
properties \eqref{eq:approx2}. The function $W$ is a polynomial whose
motivation, derivation and analysis are the subject of the remainder of
this section.

In order to use such a function $W$ for cloaking, assume $Q_0(1/z)$ is
the harmonic incident field. Then the device field needed for solving the
cloaking problem \eqref{eq:clk} is the real part of the function $U(1/z)
= -Q_0(1/z) (1-W(1/z))$ (after having undone the Kelvin transformation
we used for the analysis). The actual device field is illustrated in
Figure~\ref{fig:cloak}. On the left, a scatterer perturbs the incident
field and can be easily detected. On the right, the device field (based
on the function $W$ of Figure~\ref{fig:ensavg} is activated and
suppresses the incident field inside the cloaked region, making the
object undetectable for all practical purposes.

\subsection{A first candidate polynomial from Lagrange interpolation}
\label{sec:lag}

\begin{figure}[ht]
 \begin{center}
  \includegraphics[height=0.15\textheight]{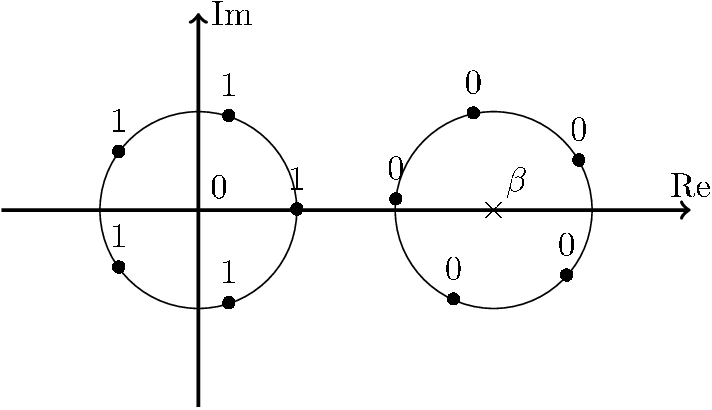}
  \hspace{1em}
  \includegraphics[height=0.15\textheight]{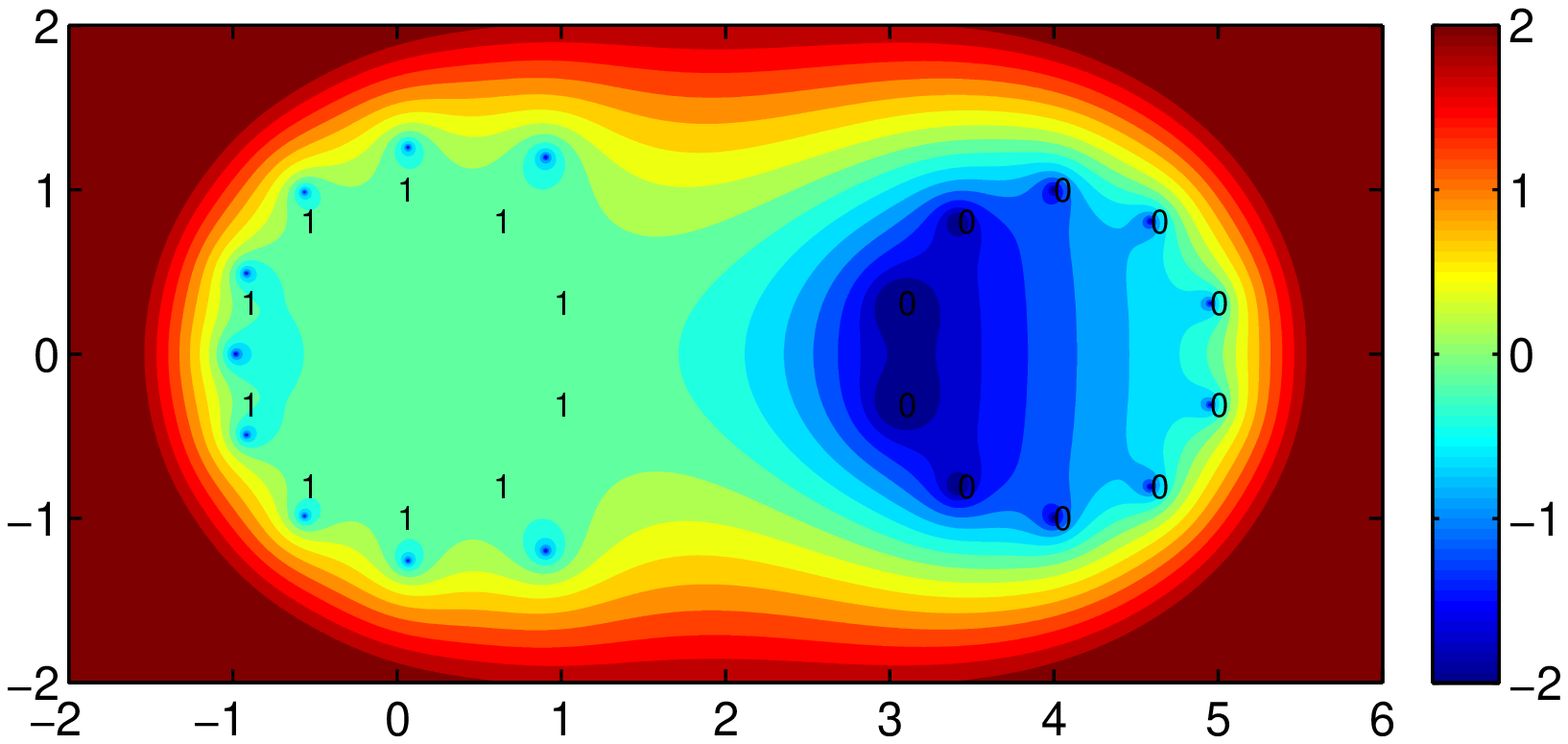}
 \end{center}
 \caption{Left: sample interpolation points for the interpolation polynomial
 $p_{\phi,\psi}$ with $n=5$, $\phi=0$ and $\psi=\pi/3$. Right: 
 the modulus of the polynomial $p_{\phi,\psi}$ with $n=10$, $\phi=-\psi=\pi/10$ and
 $\beta=4$. The color scale is logarithmic and the interpolation nodes are
 indicated by the interpolation values.}
 \label{fig:pphipsi}
\end{figure}

We present a polynomial solution to \eqref{eq:approx2} based on Lagrange
interpolation.  This is an intermediary step to motivate the explicit
solution to \eqref{eq:approx2} given later in Section~\ref{sec:her}. The
idea applies only to the case where $\alpha=R=1$ and $\beta =
p/(p^2-a^2) > 2$. The candidate solution is a polynomial that is one at
$n$ equally distributed points on $\partial B(0,1)$ and zero at $n$
equally distributed points on $\partial B(\Bc^*,1)$.  The motivation
being that by surrounding both $0$ and $\Bc^* = (\beta,0)$ by $n$ points
where the polynomial has the desired values, we hope to get close to a
polynomial satisfying \eqref{eq:approx2}. 

To be more precise, let us introduce the following family of $2n$ nodes
$\{e^{i\phi} w_j, \beta + e^{i\psi}w_j \}_{j=0}^{n-1}$. Here $\phi$ and
$\psi$ are two arbitrary angles and $w_j = \exp[2i\pi j/n]$, for
$j=0,\ldots,n-1$. Define the polynomial $p_{\phi,\psi}$ as the unique
polynomial of degree $2n-1$ satisfying,
\begin{equation}
 p_{\phi,\psi}(e^{i\phi} w_j) = 1
 ~\text{and}~
 p_{\phi,\psi}(\beta + e^{i\psi}w_j) = 0,
 ~\text{for $j = 0,\ldots,n-1$.}
\label{eq:interp1}
\end{equation}
An example of the interpolation nodes and the values of $p_{\phi,\psi}$
is shown in Figure~\ref{fig:pphipsi}(left).

The polynomial $p_{\phi,\psi}$ is unique and can be written explicitly
as
\begin{equation}
 p_{\phi,\psi}(z) = \sum_{m=0}^{n-1} q_{\phi,\psi,m}(z),
 \label{eq:pqsum}
\end{equation}
where $q_{\phi,\psi,m}(z)$ are Lagrange interpolation polynomials (see e.g.
\cite{Stoer:2002:INA}) defined for $m=0,\ldots,n-1$ by
\begin{equation}
 q_{\phi,\psi,m}(z) = 
 \Mb{\prod_{j=0,j\neq m}^{n-1} \frac{z-e^{i\phi}w_j}{e^{i\phi}w_m-e^{i\phi}w_j}}
 \Mb{\prod_{j=0}^{n-1}
 \frac{z-(\beta+e^{i\psi}w_j)}{e^{i\phi}w_m-(\beta+e^{i\psi}w_j)}},
\end{equation}
or alternatively by their interpolation properties
\begin{equation}
 q_{\phi,\psi,m}(e^{i\phi}w_j) = \delta_{mj},
 ~\text{and}~
 q_{\phi,\psi,m}(\beta+e^{i\psi}w_j) = 0,
 ~\text{for $j=0,\ldots,n-1$}.
\end{equation}
Here $\delta_{mj} = 1$ if $m=j$ and $0$ otherwise is the Kronecker
delta. Straightforward calculations give the expression
\begin{equation}
 q_{\phi,\psi,m}(z) = 
 \Mb{
 \frac{(z-\beta)^n - e^{i\psi n}}{(e^{i\phi}w_m -
 \beta)^n - e^{i\psi n}}
 }
 \Mb{
 \frac{z^n - e^{i\phi n}}{z-e^{i\phi} w_m} 
 }
 \Mb{
 \frac{1}{n(w_m e^{i\phi})^{n-1}}
 },
 \label{eq:q}
\end{equation}
which will be used later in Section~\ref{sec:her}. 

We state the following symmetry property of the polynomial 
$p_{\phi,\psi}$ for later use.
\begin{lemma}\label{lem:symm}
For any angles $\phi$ and $\psi$, the polynomial $p_{\phi,\psi}$ has the
following symmetry property: 
\begin{equation}
 p_{\phi,\psi}(z) + p_{\psi+\pi,\phi+\pi}(\beta - z) = 1.
 \label{12'}
\end{equation}
\end{lemma}
\begin{proof}
Equation \eqref{12'} follows from noticing that for $j=0,\ldots,n-1$,
\begin{equation}
 \begin{aligned}
 p_{\psi+\pi,\phi+\pi}(\beta - (\beta + e^{i\psi} w_j) ) &= 
 p_{\psi+\pi,\phi+\pi}(e^{i(\psi+\pi)} w_j) = 0,
 ~\text{and}\\
 p_{\psi+\pi,\phi+\pi}(\beta - e^{i\phi}w_j) &= 
 p_{\psi+\pi,\phi+\pi}(\beta + e^{i(\phi+\pi)}w_j)=1.
 \end{aligned}
\end{equation}
Hence the polynomial $p_{\phi,\psi}(z) + p_{\psi+\pi,\phi+\pi}(\beta - z) - 1$
must be identically zero because it is of degree $2n-1$ and has $2n$ roots
$\{e^{i\phi} w_j, \beta + e^{i\psi}w_j \}_{j=0}^{n-1}$.
\end{proof}

An actual polynomial $p_{\phi,\psi}$ is shown in
Figure~\ref{fig:pphipsi}(right). Unfortunately this polynomial is not a good
solution for problem \eqref{eq:approx2} as the regions where $p_{\phi,\psi}
\approx 1$ and $p_{\phi,\psi} \approx 0$ to within a certain tolerance
(say 1\%) are relatively small.
Changing $\phi$ and $\psi$ does not give a significant improvement. However these
polynomials are the building block for the ensemble average polynomial
solving \eqref{eq:approx2} that  we present next.

\subsection{The ensemble average polynomial}
\label{sec:her}

\begin{figure}[ht]
\begin{center}
\includegraphics[width=0.5\textwidth]{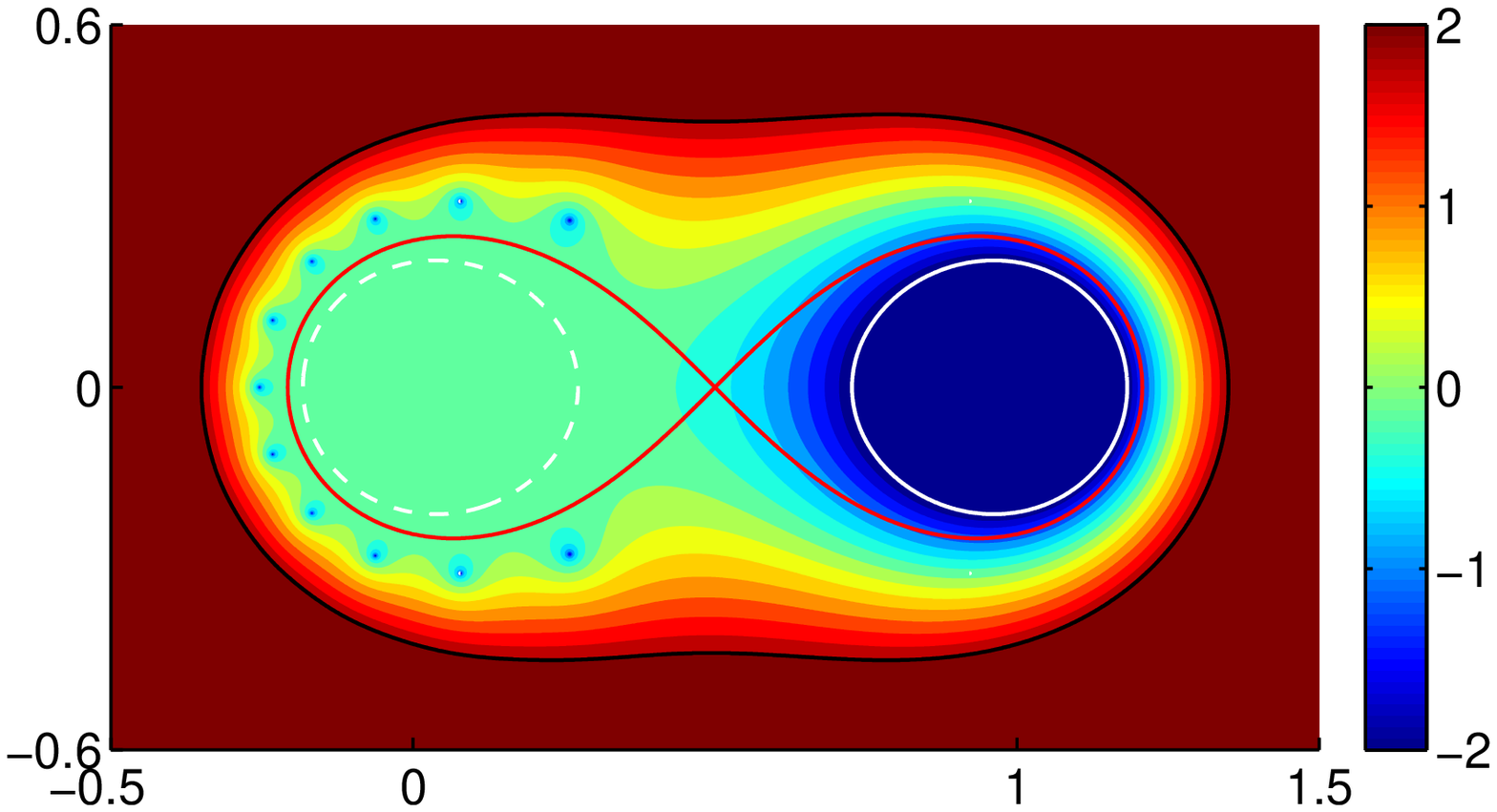}
\end{center}
\caption{The modulus of the ensemble average polynomial $\Ma{p}(z)$ for $n=12$ and
$\beta=1$. The device field used for cloaking is $\Ma{p}(1/z)$.  Within
1\% accuracy, the polynomial $\Ma{p}$ is close to one inside the dashed
white circle and close to zero inside the solid white circle. The
boundary of the convergence region $D_\beta$ of $\Ma{p}$ as $n\to\infty$
is the peanut shaped curve in red (see Theorem~\ref{Peanut-curve}). The
color scale is logarithmic from $0.01$ (dark blue) to $100$ (dark red),
with light green representing 1.} 
\label{fig:ensavg}
\end{figure}

\begin{figure}[t]
 \begin{center}
  \includegraphics[width=0.45\textwidth]{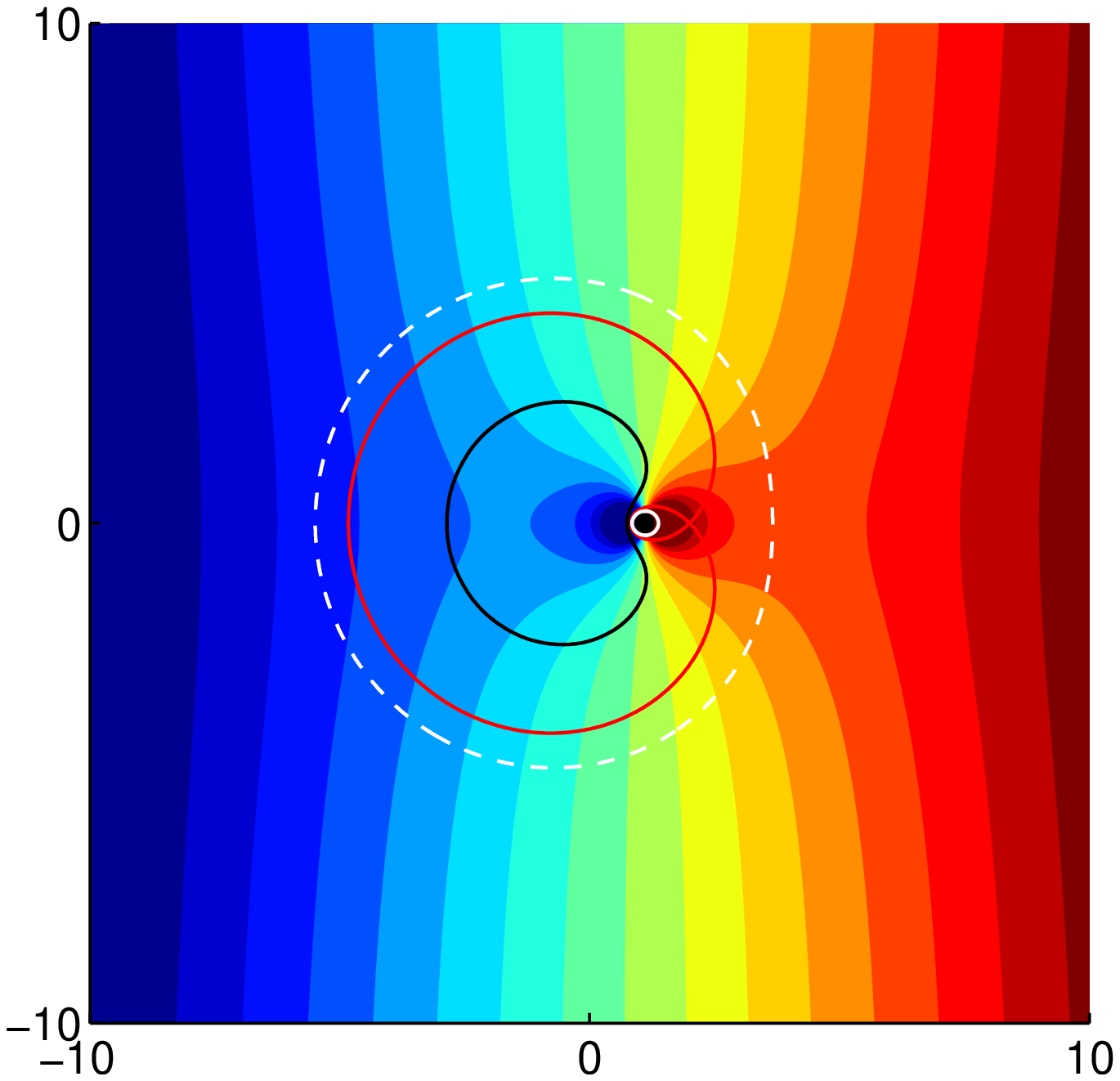}
  \hspace{1em}
  \includegraphics[width=0.45\textwidth]{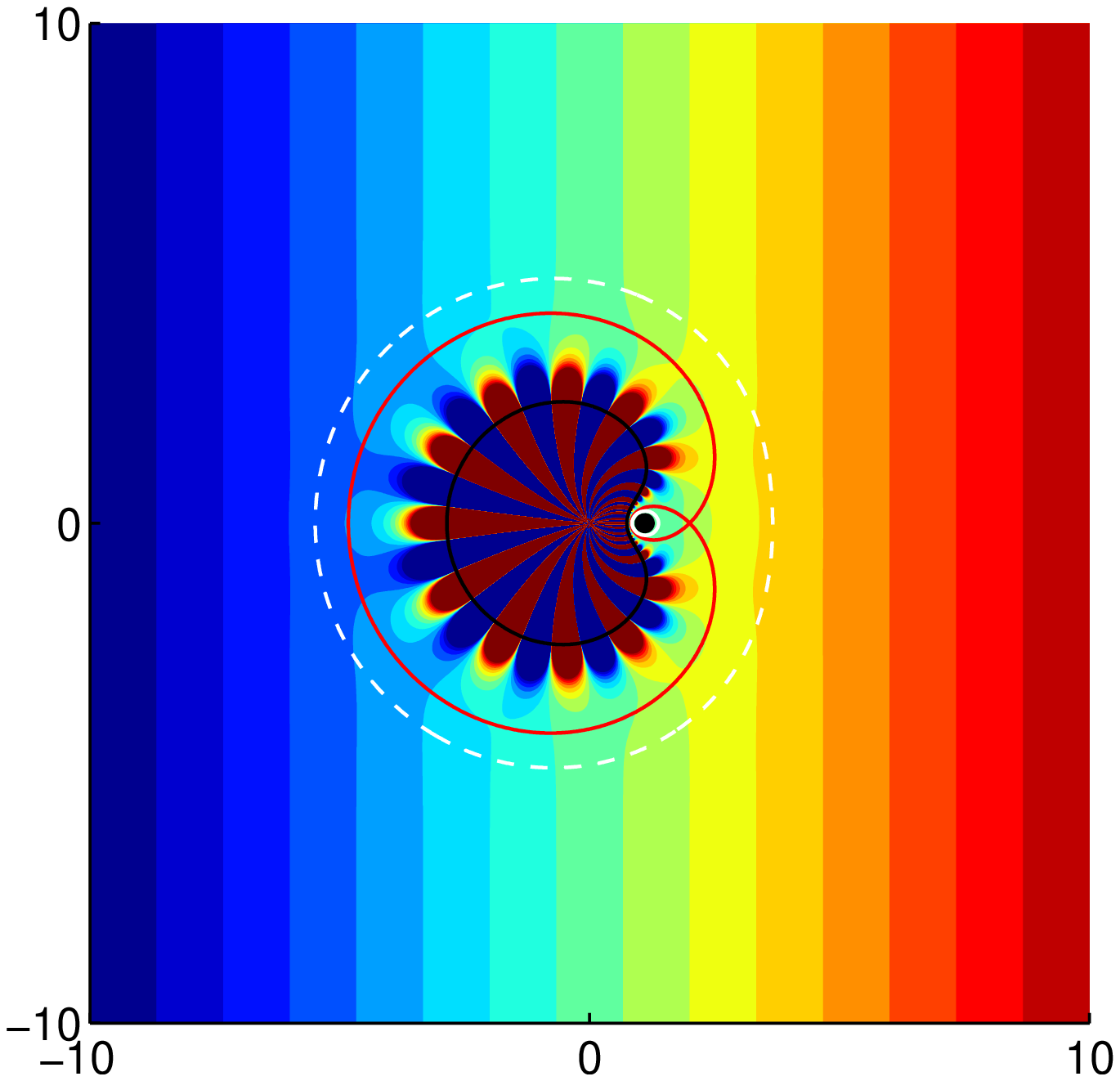}
 \end{center}
 \caption{Real part of the total field with the cloaking device active
 (right) and inactive (left), for an incident field $u_0(z) = z$ and
 $n=12$. The solid white, dashed white and red curves are the Kelvin
 transforms of the respective curves in Figure~\ref{fig:ensavg}. The
 solid black disk is an almost resonant scatterer with radius $r=0.1$,
 located at $z=1.05$ and with dielectric constant $\epsilon=-1+10^{-3}$,
 chosen to be plasmonic with a negative value close to $-1$ to amplify
 its effect.
 The solid black curve is the contour $|u|=100$. The color scale is
 linear from -10 (dark blue) to 10 (dark red).}
 \label{fig:cloak}
\end{figure}

In an effort to obtain a polynomial solution to problem
\eqref{eq:approx2} we calculate the ensemble average
of the polynomials $p_{\phi,\psi}$ with respect to the two
phase factors $\phi,\psi\in[0,2\pi]$, that is
\begin{equation}
 \Ma{p}(z) = \frac{1}{(2\pi)^2} \int_0^{2\pi} \int_0^{2\pi}
 p_{\phi,\psi} (z) d\phi d\psi.
 \label{eq:ensavg}
\end{equation}
We prove in Theorem~\ref{Peanut-curve}, using the next lemma, that indeed
$\Ma{p}$ is a solution for \eqref{eq:approx2}. An example of such polynomial
for $\beta=1$ and $n=12$ is given in Figure~\ref{fig:ensavg}.

\begin{lemma}
 The ensemble average polynomial defined in \eqref{eq:ensavg} has the
 expression
 \begin{equation}
  \Ma{p}(z) = \M{ 1-\frac{z}{\beta}}^n \, \sum_{j=0}^{n-1}
  \M{\frac{z}{\beta}}^j \frac{ (n+j-1)! } { j! (n-1)!}.
  \label{eq:ensavg:final}
 \end{equation}
\end{lemma}
\begin{proof}
 We first use the Cauchy residue theorem to compute the integral 
 \begin{equation}
  \begin{aligned}
  \frac{1}{2\pi} \int_0^{2\pi} \frac{ (z-\beta)^n - e^{i\psi n}}
  {(e^{i\phi} w_m - \beta)^n - e^{i\psi n}} d\psi &=
  \frac{1}{2i\pi} \int_{|w|=1} \frac{(z-\beta)^n - w^n}{(e^{i\phi} w_m -
  \beta)^n - w^n} \frac{dw}{w}\\
  &= \frac{(z-\beta)^n}{(e^{i\phi} w_m - \beta)^n},
  \end{aligned}
  \label{eq:int1}
 \end{equation}
 since the integrand has a single simple pole at $w=0$ in the disk
 $|w|<1$. Then by plugging \eqref{eq:int1} into the expression for
 $q_{\phi,\psi,m}$ we get that
 \begin{equation}
  \frac{1}{2\pi}\int_0^{2\pi} q_{\phi,\psi,m}(z) d\psi = 
  \frac{(z-\beta)^n}{(e^{i\phi} w_m - \beta)^n} \frac{z^n - e^{i\phi n}}{z-e^{i\phi} w_m} 
  \frac{1}{n(w_m e^{i\phi})^{n-1}}.
  \label{eq:int2}
 \end{equation}
 Recalling that $p_{\phi,\psi}$ is the sum \eqref{eq:pqsum} of $q_{\phi,\psi,m}$
 we can write
 \begin{equation}
  \Ma{p}(z) = \frac{1}{2\pi} \sum_{m=0}^{n-1} \frac{(z-\beta)^n}{(e^{i\phi} w_m - \beta)^n} \frac{z^n - e^{i\phi n}}{z-e^{i\phi} w_m} 
  \frac{1}{n(w_m e^{i\phi})^{n-1}}.
  \label{eq:int3}
 \end{equation}
 Now all $n$ terms in the previous sum are identical, therefore
 \begin{equation}
  \begin{aligned}
  \Ma{p}(z) &= \frac{1}{2\pi} \int_0^{2\pi} \frac{ (z-\beta)^n}
  {(e^{i\phi} - \beta)^n} \frac{ z^n - e^{i\phi n}}{z - e^{i\phi}}
  \frac{1}{e^{i\phi(n-1)}}\\
  & = \frac{1}{2i\pi} \int_{|w|=1} (z-\beta)^n
  \frac{z^n-w^n}{(z-w)(w-\beta)^n} \frac{dw}{w^n}\\
  & = \frac{(z-\beta)^n}{2i\pi} \int_{|w|=1} \M{ \frac{z^n}{w^n} -1}
  \frac{1}{(z-w)(w-\beta)^n} dw\\
  & = \frac{(z-\beta)^n}{2i\pi} \int_{|w|=1} \sum_{j=0}^{n-1}
  \frac{z^j}{w^{j+1} (w-\beta)^n} dw\\
  & = (z-\beta)^n \sum_{j=0}^{n-1} \frac{z^j}{j!} \frac{d^j}{dw^j} \Mb{
  \frac{1}{(w-\beta)^n} }_{w=0},
  \end{aligned}
  \label{eq:int4}
 \end{equation}
 where we used Cauchy's theorem in the last equality of \eqref{eq:int4}.
 The desired expression \eqref{eq:ensavg:final} follows by
 straightforward algebraic manipulations of \eqref{eq:int4}.
\end{proof}

\begin{remark}
 \label{rem-Hermite}
 By using elementary algebraic manipulations and \eqref{eq:int4}, it is
 possible to show that $\Ma{p}$  is the Hermite interpolation polynomial
 \cite{Stoer:2002:INA} of degree $2n-1$ that is uniquely defined by the
 $2n$ interpolation conditions
 \begin{equation}
  \Ma{p}(0) = 1, ~\Ma{p}(\beta) = 0,
  ~\text{and}~
  \Ma{p}^{(j)}(0) = \Ma{p}^{(j)}(\beta) = 0,
  ~\text{for $j=1,\ldots,n-1$.}
  \label{eq:herm}
 \end{equation}
\end{remark}

Notice that the ensemble average polynomial inherits the symmetry
property \eqref{12'} for $p_{\phi,\psi}$, that is
\begin{equation}
 \Ma{p}(z) + \Ma{p}(\beta-z) = 1.
 \label{eq:symm2}
\end{equation}
This symmetry property means that by design, the polynomial gives as
good an approximation to one near the origin as the approximation to
zero near $\beta$.

\subsection{Asymptotics of the ensemble average polynomial}
\label{sec:conv}
We now study the behavior of the polynomial $\Ma{p}$ (defined in
\eqref{eq:ensavg}) as $n \to \infty$. The following result shows that
the polynomial $\Ma{p}$ solves the problem \eqref{eq:approx2}, and gives
limits to the size of the cloaked region.

\begin{theorem}
\label{Peanut-curve} 
The ensemble average polynomial $\Ma{p}$ can be written as
\begin{equation}
 \Ma{p} = \frac{1}{2} + \sum_{k=0}^{n-1} \frac{(2k)!}{(k!)^2}
 \M{ \frac{z}{\beta} \M{1-\frac{z}{\beta}}}^k
 \M{ \frac{1}{2} - \frac{z}{\beta}}.
 \label{eq:ensavg:conv}
\end{equation}
The polynomial $\Ma{p}(z)$ converges as $n\to \infty$ if and only if $z$
belongs to the convergence region
\begin{equation}
 D_\beta = \Mcb{z \in \complex, ~ |z^2 - \beta z| < \frac{\beta^2}{4}}.
  \label{eq:dbeta}
\end{equation}
The convergence is uniform on compact subsets of $D_\beta$ to the
function 
\begin{equation}
 \chi (z) = \begin{cases} 1  & \text{if}\;\;\Re(z) <
 \beta/2, \\ 0 & \text{otherwise.} \end{cases}
 \label{eq:chi}
\end{equation}
For large enough $n$, the polynomial $\Ma{p}$ 
solves \eqref{eq:approx2} if and only if
\begin{equation}
 \frac{1}{R} < \frac{\beta}{2\sqrt{2}+2}
 ~~~\text{and}~~~
 \alpha < \frac{\beta}{2\sqrt{2}+2}.
 \label{eq:constraints}
\end{equation}
\end{theorem}
\begin{proof}
Consider the function 
\begin{equation}
f_n(t) = (1-t)^n \sum_{j=0}^{n-1} t^j {n+j-1 \choose j}
\label{22}
\end{equation}
where for any positive integers $m$ and $p$, ${m\choose
p}=\frac{m!}{p!(m-p)!}$. Note that, from \eqref{eq:ensavg:final} we have 
\begin{equation} 
 f_n(t)=\Ma{p}(\Gb t), \mbox { for all } t\in \complex.
\label{22'}
\end{equation}
Then for all $t\neq 1$ we obtain,
\begin{equation}
\begin{aligned}
\frac{f_{n+1}(t)}{(1-t)^{n+1}}&=\sum_{j=0}^{n}t^j{n+j\choose j}\\
 &=1+\sum_{j=1}^{n}t^j\left[{n+j-1\choose
j-1}+{n+j-1\choose j}\right]\\
&=\sum_{j=1}^{n}t^j{n+j-1\choose
j-1}+\left[1+\sum_{j=1}^{n}t^j{n+j-1\choose
j}\right]\\
&=\sum_{k=0}^{n-1}t^{k+1}{n+k\choose
k}+\sum_{j=0}^{n}t^j{n+j-1\choose j}\\
&=t\left[\frac{f_{n+1}(t)}{(1-t)^{n+1}}-t^n{2n\choose
n}\right]+\frac{f_{n}(t)}{(1-t)^n}+t^n{2n-1\choose n}.
\end{aligned}
\label{23}
\end{equation}
In the above equation we used the recurrence relation
\[
 {m\choose p}={m-1\choose p-1}+{m-1\choose p}, \mbox{ for any
integers } m,p>0.
\]
From \eq{23}, for any integer $n\geq 1$ we obtain, 
\begin{equation}
\begin{aligned}
f_{n+1}(t)&=f_n(t)-(1-t)^n t^{n+1}{2n\choose
n}+(1-t)^n t^n{2n-1\choose n}\\
&=f_n(t)-(1-t)^n t^n{2n\choose n}\M{t-\frac{1}{2}}, \;\;\mbox{
for all } t\neq 1.
\end{aligned}
\label{24}
\end{equation}
In \eq{24} we used the identity
\[
 {2n\choose n}=2{2n-1\choose n}, \mbox{ for every integer
}n\geq 1.
\]
From the first order linear recurrence \eq{24} we obtain, 
\begin{equation}
f_n(t)=
\frac{1}{2}+\sum_{k=0}^{n-1}(1-t)^{k}t^{k}{2k\choose
k}\left(\frac{1}{2}-t\right)
\label{27}
\end{equation}
and this is valid for all $n\geq 1$ and all $t\in\mathbb C$ (as \eq{27}
which was initially obtained for $t\neq 1$ checks also for $t=1$). The
final expression \eqref{eq:ensavg:conv} follows from substituting
$t=z/\Gb$ in \eq{27} and using \eq{22'}.

Notice that the polynomial $\Ma{p}$ is in fact the $n$-th order partial
sum of the following infinite sum,
\[
\frac{1}{2}+\sum_{k=0}^{\infty}(1-\frac{z}{\Gb})^{k}\left(\frac{z}{\Gb}\right)^{k}{2k\choose
k}\left(\frac{1}{2}-\frac{z}{\Gb}\right).
\]
By the ratio test this series converges uniformly on compacts subsets of
the region $D_\Gb$ (defined at \eq{eq:dbeta}) to a limit function
$\varphi$ and diverges in $\complex \setminus {\overline D_\Gb}$.
From the uniform convergence of $\Ma{p}$ we
deduce the analyticity of $\varphi$ in $D_\beta$ and by using the Taylor
expansion around the origin for $\varphi$, the Remark \ref{rem-Hermite},
and the symmetry property \eqref{eq:symm2} we obtain convergence to the
function \eqref{eq:chi} inside $D_\beta$.

We now study the convergence region $D_\beta$ in order to show that the
constraints \eqref{eq:constraints} are necessary and sufficient for
$\Ma{p}$ to solve \eqref{eq:approx2}. First notice that the definition
of the region $D_\beta$ and simple algebra reveal that
\begin{eqnarray}
 \frac{1}{R}< \frac{\Gb}{2\sqrt{2}+2} &\Leftrightarrow&
  \frac{1}{R} e^{-i\pi} \in (D_\Gb\cap\{z\in{\mathbb C},
 2Re(z)<\Gb\}),~\text{and}\label{41'}\\
 \Ga < \frac{\Gb}{2\sqrt{2}+2} &\Leftrightarrow& 
 \Gb+\Ga \in(D_\Gb\cap\{z\in{\mathbb C}, 2Re(z)>\Gb\}).\label{41}
\end{eqnarray}
Next we show that
\begin{eqnarray}
\frac{1}{R}< \frac{\Gb}{2\sqrt{2}+2} &\Leftrightarrow& 
 B(0,1/R)\Subset (D_\Gb\cap\{z\in{\mathbb C},\;
 2Re(z)<\Gb\}),~\text{and}\label{28''}\\
 \Ga < \frac{\Gb}{2\sqrt{2}+2} &\Leftrightarrow&
 B(\Bc^*,\alpha) \Subset (D_\Gb\cap\{z\in{\mathbb C},\;
2Re(z)>\Gb\}),\label{28} 
\end{eqnarray} 
where $\Subset$ is the classical symbol for compact inclusions. By using
the equivalences \eqref{41'} and \eqref{41} it is easy to check that the
inclusions in \eqref{28''} and \eqref{28} imply the constraints
\eqref{eq:constraints}.  To show the implication ($\Rightarrow$), we
first show that for any two positive real numbers $l,q$ with
$2\max\{l,q\}<\beta$, we have 
\begin{eqnarray}
 l e^{-\pi i}\in D_\Gb 
 &\Leftrightarrow& 
 \overline{B(0,l)} \Subset (D_\Gb\cap\{z\in\complex,
 2Re(z)<\Gb\})~\text{and}\label{30}\\
 \Gb+q \in D_\Gb 
 &\Leftrightarrow& 
 \overline{B(\Bc^*,q)} \Subset 
 (D_\Gb\cap\{z\in\complex, 2Re(z)>\Gb\}).
 \label{31}
\end{eqnarray}
Let us first show the equivalence \eqref{30}. The sufficiency
($\Leftarrow$) is immediate. For the other implication ($\Rightarrow$),
we can use the definition of $D_\Gb$ to show that for any
$\theta\in [-\pi, \pi]$ we have, 
\begin{eqnarray}
le^{i\theta}\in D_\Gb 
&\Leftrightarrow& 
|l^2e^{i\theta}-\Gb l|<\frac{\Gb^2}{4}\nonumber\\
&\Leftrightarrow& 
(l^2e^{i\theta}-\Gb l)(l^2e^{-i\theta}-\Gb l)<\frac{\Gb^4}{16}\nonumber\\
&\Leftrightarrow&
l^4+\Gb^2 l^2-2l^3\Gb\cos\theta-\frac{\Gb^4}{16}<0.
\label{33}
\end{eqnarray}
Since we assumed $l e^{-\pi i}\in D_\Gb$, equation \eq{33} immediately
implies that
\begin{equation}
l^4+\Gb^2 l^2+2l^3\Gb-\frac{\Gb^4}{16}<0.
\label{32} 
\end{equation} 
Consider the even function $f:[-\pi,\pi]\rightarrow \real$ defined by, 
\begin{equation}
f(\theta)= l^4+\Gb^2
l^2-2l^3\Gb\cos\theta-\frac{\Gb^4}{16}.
\label{34}
\end{equation}
Observe now that, because $l>0$, its derivative $f'(\theta)=2l^3\Gb
\sin\theta$ has the signs, 
\begin{equation}
f'(\theta) \geq 0 ~\text{for $\theta\in [0,\pi]$}
~~\text{and}~~~
f'(\theta) \leq 0 ~\text{for $\theta\in [-\pi, 0]$}.
\label{35}
\end{equation}
Note that from inequality \eq{32} and the definition \eq{34} of $f(\theta)$
one immediately obtains 
\begin{equation}
f(-\pi)=f(\pi)<0.
\label{36}
\end{equation}
Then the signs of $f'(\theta)$ in \eq{35} together with the particular
values of $f(\theta)$ in \eq{36} imply
\begin{equation}
f(\theta)<\max\{f(\pi),f(-\pi)\}<0,\mbox{ for all
}\theta\in[-\pi,\pi].
\label{37}
\end{equation}
Because of equivalence \eqref{33} we conclude from inequality \eqref{37}
that
\begin{equation}
le^{i\theta}\in D_\Gb,\mbox{ for any } \theta\in
[-\pi,\pi].
\label{38}
\end{equation}
From the conditions on $l$, we have that $2Re(le^{i\theta})\leq 2l<\Gb$
and by using this in \eq{38} we obtain 
\begin{equation}
le^{i\theta}\in (D_\Gb\cap\{z\in\complex,
2Re(z)<\Gb\}),\mbox{ for any } \theta\in [-\pi,\pi].
\label{39}
\end{equation}
Inclusion \eqref{39} together with the convexity of
$D_\Gb\cap\{z\in\complex, 2Re(z)<\Gb\}$ implies that 
\[
 {\overline B}_l(0)\Subset D_\Gb\cap\{z\in{\mathbb C},
 2Re(z)<\Gb\}.
\] 
This establishes the equivalence \eqref{30}. From the definition of
the set $D_\Gb$, by simple algebraic manipulation we obtain
\begin{equation}
\Gb+q \in D_\Gb \Leftrightarrow qe^{-\pi i}\in D_\Gb.
\label{40}
\end{equation}
Equivalence \eq{40} clearly implies that \eq{31} follows from \eq{30}
applied to $q$ instead of $l$. Finally, observing the fact that
the constraints \eq{eq:constraints} imply $2\max\{\frac{1}{R},\Ga\}<\Gb$ and using
equivalences \eq{41'}, \eq{41}, \eq{30} and \eq{31} for $\frac{1}{R}$
and $\Ga$ instead of $l$ and $q$ respectively, we obtain the desired
equivalences \eq{28''} and
\eq{28}. By using the uniform convergence of the polynomial $\Ma{p}$ to the
function $\chi(z)$ in $D_\beta$, and equivalences
\eq{28''} and \eq{28} we obtain that the constraints
\eqref{eq:constraints} are indeed necessary and sufficient for
convergence of $\Ma{p}$.
\end{proof}

\begin{remark} 
The expression \eqref{eq:ensavg:conv} of the ensemble
average polynomial could also be obtained by generalizing to
distributions a theorem by Ramharter \cite{Ramharter:1991:RHL} (which is
in turn a generalization of a result due to Berger and Tasche
\cite{Berger:1988:HLI}). To remain concise, we prefer to include a
direct proof.
\end{remark}

\section{Summary} 
For the Laplace equation we have shown the existence
of a device capable of cloaking a region exterior to the device,
assuming a priori knowledge of the incident field. The proof relies on a
non-constructive harmonic function approximation result. The theory does
not constrain the size and relative positions of the device and  cloaked
region, as long as they are bounded, disjoint and the complement of
their union is connected. Although the construction of such a cloaking
device is clearly not unique, we presented earlier in
\cite{Guevara:2009:AEC} a construction based on an explicit polynomial.
Here we rigorously justify this construction and show that the
constraints \eqref{eq:constraints} must be satisfied in order to have a
proper active exterior cloak. Because of the constraints
\eqref{eq:constraints}, the current strategy fails to cloak large
objects ($\Ga$ large) unless they are sufficiently far from the origin
($\Gb$ large enough). In \cite{Guevara:2011:TEA} (see Conjecture 1), we
present without proof, as a conjecture, an extension of
Theorem~\ref{Peanut-curve} which gives a wider choice of cloaks and that
is supported by numerical experiments.

\begin{acknowledgements}
The authors are grateful for support from the National Science
Foundation through grant DMS-070978 and express their gratitude to
Robert V. Kohn, Jeffrey Rauch and John Willis for insightful
suggestions. The paper was in part written while the authors were
visiting the Mathematical Sciences Research Institute for the Inverse
problems program during the Fall semester of 2010. FGV was partially
supported through the National Science Foundation grant DMS-0934664.
\end{acknowledgements}

\bibliographystyle{abbrvnat}
\bibliography{ecbib}

\end{document}